\documentclass[letterpaper,10pt,conference]{ieeeconf}  

\IEEEoverridecommandlockouts                              
\usepackage{mathptmx} 
\usepackage{times} 
\usepackage{amsmath} 
\usepackage{amssymb}  

\usepackage[noadjust,sort]{cite}

\usepackage{array,arydshln}

\usepackage{url,calc}

\title{\LARGE \bf Spreading Processes over Socio-Technical Networks\\with Phase-Type Transmissions}

\author{Masaki Ogura and Victor M.~Preciado
\thanks{The authors are with the Department of Electrical and Systems
Engineering, University of Pennsylvania, Philadelphia, PA 19014, USA.
Email:  {\tt\small \{ogura,preciado\}@seas.upenn.edu}}%
\thanks{{This work was supported in part by the NSF under grants CNS-1302222
and IIS-1447470.}}%
}

\newtheorem{definition}{Definition}[section]
\newtheorem{assumption}[definition]{Assumption}
\newtheorem{lemma}[definition]{Lemma}
\newtheorem{proposition}[definition]{Proposition}
\newtheorem{theorem}[definition]{Theorem}
\newtheorem{remark}[definition]{Remark}

\newtheorem{example}[definition]{Example}

\newcommand{\smallbmatrix}[1]{{\begin{bsmallmatrix}#1\end{bsmallmatrix}}}

\DeclareSymbolFont{bbold}{U}{bbold}{m}{n}
\DeclareSymbolFontAlphabet{\mathbbold}{bbold}
\newcommand{\onev}{\mathbbold{1}}
\newcommand{\onevtop}{\onev^{\!\!\top}\!}

\usepackage{calrsfs}
\DeclareMathAlphabet{\pazocal}{OMS}{zplm}{m}{n}
\renewcommand{\mathcal}[1]{\pazocal{#1}}

\def\usetodonotes{}	
\ifdefined\usetodonotes
\def\disabletodonotes{}
\else
\def\disabletodonotes{disable}
\fi

\usepackage[usenames]{xcolor}
\definecolor{darkgreen}{rgb}{0.0, 0.2, 0.13}
\definecolor{darkspringgreen}{rgb}{0.09, 0.45, 0.27}
\definecolor{darkcandyapplered}{rgb}{0.64, 0.0, 0.0}
\usepackage[\disabletodonotes,textwidth=2.475cm,
textsize=footnotesize,backgroundcolor=white,
linecolor=darkcandyapplered,bordercolor=darkcandyapplered]{todonotes}

\def\usetodonotes{}

\usepackage{mathtools}
\mathtoolsset{showonlyrefs=true}
\mathtoolsset{centercolon}

\usepackage{flushend}

\newcommand{\afterequation}{\vskip 3pt}

\allowdisplaybreaks[4]

\begin{document}

\maketitle
\thispagestyle{empty}
\pagestyle{empty}

\begin{abstract}
Most theoretical tools available for the analysis of spreading processes over networks assume exponentially distributed transmission and recovery times. In practice, the empirical distribution of transmission times for many real spreading processes, such as the spread of web content through the Internet, are far from exponential. To bridge this gap between theory and practice, we propose a methodology to model and analyze spreading processes with arbitrary transmission times using phase-type distributions. Phase-type distributions are a family of distributions that is dense in the set of positive-valued distributions and can be used to approximate any given distributions. To illustrate our methodology, we focus on a popular model of spreading over networks: the susceptible-infected-susceptible (SIS) networked model. In the standard version of this model, individuals informed about a piece of information transmit this piece to its neighbors at an exponential rate. In this paper, we extend this model to the case of transmission rates following a phase-type distribution. Using this extended model, we analyze the dynamics of the spread based on a vectorial representations of phase-type distributions. We illustrate our results by analyzing spreading processes over networks with transmission and recovery rates following a Weibull distribution.
\end{abstract}

\section{Introduction} Understanding and controlling spreading processes over
complex networks is an important problem with applications in many relevant
fields, including public health~\cite{Pastor-Satorras2015}, malware
spreading~\cite{Garetto2003}, and information propagation over socio-technical
networks~\cite{Lerman2010}. One fundamental result on the analysis of spreading
processes over networks is the close connection between the spectral radius of
the network and the dynamics of the
spread~\cite{Lajmanovich1976,Chakrabarti2008,VanMieghem2009a,Sahneh2012a}. Based
on this result, the authors
in~\cite{Preciado2013,Preciado2013a,Preciado2013b,Preciado2014} proposed an
optimization framework to find the optimal allocation of resources to control
epidemic outbreak in different scenarios.

The vast majority of spreading models over networks assume exponentially
distributed transmission and recovery rates. In contrast, empirical observations
indicate that most real-world spreading processes do not satisfy this
assumption~\cite{Limpert2001,Lloyd2001a,Lloyd2001}. For example, the transmission rates of
human immunodeficiency viruses present a distribution far from
exponential~\cite{Blythe1988}. In the context of socio-technical networks, the
inter-arrival time of Twitter messages or the propagation time of news stories
on a social media site can be explained using lognormal
distributions~\cite{Lerman2010,Limpert2001,Mieghem2011a,Doerr2013}.

There are only a few results available for analyzing spreading processes over networks with non-exponential transmission (and/or recovery) rates. The experimental study in~\cite{VanMieghem2013} confirmed the drastic effect that non-exponential rates can have on the speed of spreading. In~\cite{Jo2014}, an analytically solvable (although rather simplistic) model of spreading with non-exponential rates was proposed. An approximate analysis of spreading processes over graphs with general transmission and recovery times was proposed in~\cite{Cator2013a} using asymptotic approximations.

In this paper, we  propose an alternative approach to analyze general
transmission and recovery rates using phase\nobreakdash-type
distributions~\cite{Asmussen1996}. In contrast with \cite{Cator2013a}, we
provide an analysis of general infection and recovery times for phase-type
distributions without relying on asymptotic approximations. The class of
phase-type distributions is dense in the space of positive-valued
distributions~\cite{Cox1955}, hence, we can theoretically analyze arbitrary
transmission and recovery rates. Furthermore, there are efficient algorithms to
compute the parameters of a phase-type distributions to approximate any given
distribution \cite{Asmussen1996}. To validate our approach, we verify that the
approximations in \cite{Cator2013a} are valid under certain irreducibility
assumptions. The key tool used in our derivations is a vectorial representations
proposed in~\cite{Brockett2008a}, which we use to represent phase-type
distributions.

The paper is organized as follows. In Section~\ref{sec:phase}, we state the spreading model under consideration. In Section~\ref{sec:inf=exp}, we analyze this model when the recovery times follow a phase-type distribution, while the transmission times follow an exponential distribution. Section \ref{sec:cure=exp} is devoted to the dual case when the transmission times follow a phase-type distribution, while the recovery times follow an exponential distribution. Numerical simulations are presented in Section~\ref{sec:num}.

\subsection{Mathematical Preliminaries}

An undirected graph is a pair~$\mathcal G = (\mathcal V, \mathcal E)$, where $\mathcal V = \{1, \dotsc, n\}$ is the set of nodes, and $\mathcal E$ is the set of edges, consisting of distinct and unordered pairs~$\{i, j\}$ for $i, j\in \mathcal V$. We say that $i$ is a neighbor of~$j$ (or that $i$ and $j$ are adjacent) if $\{i, j\} \in \mathcal E$. The adjacency matrix~$A\in \mathbb{R}^{n\times n}$ of $\mathcal G$ is defined as the $\{0, 1\}$\nobreakdash-matrix whose $(i,j)$ entry is one if $i$ and $j$ are adjacent, $0$ otherwise. Finally, the expectation of a random variable is denoted by~$E[\cdot]$.

We let $I$ and $O$ denote the identity and zero matrices with appropriate dimensions. Let $e_i^p$ denote the $i$th standard unit vector in $\mathbb{R}^p$ and define $E_{ij}^{p} = e_i^{p}(e_j^{p})^\top$. By $\onev_p$ we denote the $p$-vector whose entries are all one. We omit the dimension~$p$ when it is obvious from the context. A real matrix $A$, or a vector as its special case, is said to be nonnegative (positive), denoted by $A\geq 0$ ($A>0$, respectively), if $A$ is nonnegative (positive, respectively) entry-wise. The notations $A\leq 0$ and $A<0$ are understood in the obvious manner. We denote the Kronecker product of matrices~$A$ and $B$ by $A\otimes B$. Let $A$ be a square matrix. The maximum real part of the eigenvalues of $A$ is denoted by~$\eta(A)$. We say that $A$ is Hurwitz stable if $\eta(A)< 0$. Also, we say that $A$ is Metzler if the off-diagonal entries of $A$ are all non-negative. Finally, $A$ is said to be irreducible if no similarity transformation by a permutation matrix makes $A$ into block upper triangular matrix.

Below we state some basic lemmas about Metzler matrices. The first lemma about the Hurwitz stability of Metzler matrices is standard and thus its proof is omitted.

\begin{lemma}[\cite{Farina2000}] \label{lem:Metzler:Hurwitz}
For a Metzler matrix $A$, the following conditions are equivalent. 
\begin{enumerate}
\item $A$ is Hurwitz stable. 

\item There exists a positive vector $v$ such that $Av < 0$. 

\item $A$ is nonsingular and $A^{-1} \leq 0$. 
\end{enumerate}
\afterequation
\end{lemma}

The next lemma characterizes the marginal case when $\eta(A) = 0$.

\begin{lemma}\label{lem:metzler}
Let $A$ be an irreducible Metzler matrix. Then $\eta(A) = 0$ if and only if $A$ has a positive null vector.
\end{lemma}

\begin{proof}
In the proof we denote the spectral radius of a matrix by $\rho(\cdot)$. Notice
that, since $A$ is Metzler, there exists $c\geq 0$ such that $B = A+cI$ is
positive and $\eta(B) = \rho(B)$. First assume that $\eta(A) = 0$. Then
the positive and irreducible matrix $B$ satisfies $\rho(B) =  c$. Therefore, by
the Perron-Frobenius theorem (see, e.g., \cite{Vandergraft1968}), there exists a
positive vector $v$ such that $Bv = cv$, which implies $Av = 0$. Hence $v$ is a
positive null vector of $A$. The other direction can be proved in a similar way
and hence it is omitted.
\end{proof}

Finally we state the following lemma. 

\begin{lemma}\label{lem:eta<...}
Let $A, B\in \mathbb{R}^{n\times n}$. If $A$ is irreducible and Metzler and also $B \neq 0$ is non-negative, then $\eta(A) < \eta(A+B)$.
\end{lemma}

\begin{proof}
The inequality holds true if $A$ is non-negative by \cite[Theorem~4.6]{Vandergraft1968}. If $A$ is not non-negative, we consider instead $A+cI$ where $c\geq 0$ is a scalar such that $A+cI$ is nonnegative and also $\eta(A+cI) = \rho(A+cI)$. The details of the proof are omitted.
\end{proof}

\section{SIS Model with Phase-type Transmission\\and Recovery
Times}\label{sec:phase}

In this paper, we investigate a popular model of spreading
over networks called the susceptible-infected-susceptible (SIS) model. We provide a definition of this model as a family of continuous-time stochastic process below:

\begin{definition}
Let $\mathcal G$ be an undirected graph with $n$ nodes. We say that a stochastic process $\{z_i(t)\}_{t\geq 0}$, ($i= 1, \dotsc, n$) taking values in the set $\{\text{susceptible}, \text{infected}\}$ is a \emph{susceptible-infected-susceptible (SIS) model over $\mathcal G$} if the process satisfies the following conditions:
\begin{enumerate}
\item[D1)] For all $t\geq 0$ and $i\in \{1, \dotsc, n\}$, there exists a random number $R^{(i)}(t) > 0$ such that, once $z_i$ becomes infected at time~$t$, it remains infected during the time interval $[t, t+R^{(i)}(t)]$ and becomes susceptible at time~$t + R^{(i)}(t)$. We call $R^{(i)}(t)$ a \emph{recovery time}.

\item[D2)] For all $t$ and $i$, there exists a renewal process $0 = T_0^{(i)}(t) < T_1^{(i)}(t) < \cdots$ such that, for every neighbor $j$ of $i$ and $k \geq 1$, if $T^{(i)}_k(t) < R^{(i)}(t)$, then $z_j$ becomes infected at time~$t+T^{(i)}_k(t)$. We call the numbers $T_k^{(i)}(t)$ ($k\geq 1$) \emph{transmission times}.

\item [D3)] There exists a subset $\Lambda \subset \{1,\dotsc, n\}$ such that $z_i(0)$ is infected if $i\in \Lambda$ and $z_i(0)$ is susceptible otherwise. We regard that these initially infected nodes become infected at time~$0$, i.e., a node $i\in \Lambda$ has the recovery time~$R^{(i)}(0)$ and the transmission times $T_k^{(i)}(0)$.
\end{enumerate}
We say that the infection-free equilibrium $p_i(t) = \Pr(\text{$i$ is infected
at time~$t$}) \equiv 0$ of the SIS model is \emph{exponentially stable} if,  as $t\to\infty$,  $p_i(t)$ ($i=1, \dotsc, n$) converges to zero exponentially fast.\footnote{For simplicity in our presentation, we often say that $i$ is infected (or susceptible) at time~$t$ when $z_i(t) = \text{infected}$ (susceptible, respectively).}
\end{definition}

Throughout the paper we assume that all the recovery times and transmission times are independent to each other and, moreover, depend neither on $t$ nor $i$. Therefore, we hereafter omit $t$ and $i$ and write $R^{(i)}(t)$ and $T_k^{(i)}(t)$ as $R$ and $T_k$ when no confusion arises.

In \cite{Cator2013a}, the authors present the following two criteria for the
stability of the infection-free equilibrium based on an asymptotic argument.
First they show that, under the assumption that the inter-renewal times of $T$
follow an exponential distribution with mean $1/\beta$, if
\begin{equation}\label{eq:Cator:PHrecovery}
\eta(A) < {1}/(\beta E[R]), 
\end{equation}
then the meta-stable state of the infection probabilities must be equal to zero.
 They then show that, under the assumption that $R$ follows an exponential
distribution with mean $1/\delta$,  if
\begin{equation}\label{eq:Cator:PHinfection}
\eta(A) < \left(1-f(\delta)\right)/{f(\delta)}, 
\end{equation} 
where $f$ is the Laplace transform of the probability density function of the
inter-renewal times of $T$, then the meta-stable state of the infection
probabilities must be equal to zero. For the precise meaning of the
meta-stability, the readers are referred to~\cite{Cator2013a}.

One of the primary purposes of this paper is to give justifications for the above claims, without asymptotic approximations, under the assumption that either transmission or recovery times follow a phase-type distribution~\cite{Asmussen1996} introduced below. Consider a time-homogeneous Markov process in continuous-time with $p+1$ ($p\geq 1$) states such that the states~$1$, $\dotsc$, $p$ are transient and the state $p+1$ is absorbing. The infinitesimal generator of the process is then necessarily of the form
\begin{equation}\label{eq:infgen}
\begin{bmatrix}
S & v\\
0 & 0
\end{bmatrix}, \ v = -S\onev, 
\end{equation}
where $S \in \mathbb{R}^{p\times p}$ is an invertible Metzler matrix with non-positive row-sums. Let {$\smallbmatrix{\phi\\0} \in \mathbb{R}^{p+1}$} ($\phi \in \mathbb{R}^p$) denote the initial distribution of the Markov process. Then, the time to absorption into the state~$p+1$, denoted by $(\phi, S)$, is called a \emph{phase-type distribution}. It is known that the set of phase-type distributions is dense in the set of positive valued distributions~\cite{Cox1955}. Moreover, there are efficient fitting algorithms to approximate a given arbitrary distribution by a phase-type distribution  \cite{Asmussen1996}.

In order to analyze the conditions \eqref{eq:Cator:PHrecovery} and
\eqref{eq:Cator:PHinfection}  using phase-type distributions, we shall state the
following two assumptions:

\begin{assumption}\label{assm:Tbeta}
$R$ follows a phase-type distribution~$(\phi, S)$ and the
inter-renewal times of $T$ follow an exponential distribution with
mean $1/\beta$.
\end{assumption}

\begin{assumption}\label{assm:Rdelta}
The inter-renewal times of $T$ follow a phase-type distribution
$(\phi, S)$ and $R$ follows an exponential distribution with mean
$1/\delta$.
\end{assumption}

\newcommand{\dN}{\,dN}
\subsection{Vector Representations of Phase-type Distributions}

Vector representations of phase-type distributions play a crucial role
in our analysis of SIS models. In order to introduce
these representations, we first recall a vector representation of
time-homogeneous Markov processes introduced in \cite{Brockett2008a}:

\begin{lemma}[{\cite[Section~II]{Brockett2008a}}]\label{lem:brockett}
Let $Q$ be the infinitesimal generator of a time-homogeneous Markov
process taking its values in $\{1, \dotsc, p\}$. For distinct $i$ and
$j$ let $N_{Q_{ij}}$ denote the Poisson counter of rate $q_{ij}$.
Consider the stochastic differential equation 
\begin{equation*}
dx = \sum_{\ell, m=1}^p
(E_{m \ell} - E_{\ell \ell})x \dN_{q_{\ell m}}
\end{equation*}
with the initial state
$x(0)$ being a standard unit vector in~$\mathbb{R}^p$ with probability
one. Then $x$ is the time-homogeneous Markov process with the
infinitesimal generator $Q$ and the state space~$\{e_1, \dotsc,
e_p\}\subset \mathbb{R}^p$.
\end{lemma}

Using this lemma, we now provide a vector representation of phase-type
distributions, as follows:

\begin{lemma}\label{lem:PHdist}
Let $(\phi, S)$ be a phase-type distribution. Let $e_\phi$ denote the
probability distribution on the standard unit vectors in
$\mathbb{R}^p$ given by $\Pr(e_\phi = e_i) = \phi_i$ for $i=1, \dotsc,
p$. Consider the stochastic differential equation
\begin{equation}\label{eq:PHnonrenewal}
dx 
= 
\sum_{\ell, m=1}^p(E_{m\ell} - E_{\ell\ell}) x\dN_{S_{\ell m}} 
- 
\sum_{\ell=1}^p E_{\ell\ell} x \dN_{v_\ell}
\end{equation}
with the initial condition that $x(0)$ follows $e_\phi$. Then, the extinction
time random variable
\begin{equation}\label{eq:defR}
\begin{aligned}
\hspace{-1cm}R 
&= 
\inf\{t > 0: x(t) = 0\}
\\
&=
\inf\{t > 0:\exists \ell^*,\,x(t)=e_\ell^* \mbox{ and $N_{v_{\ell^*}}$ jumps at time~$t$}\}
\end{aligned}
\end{equation}
follows $(\phi, S)$.
\end{lemma}

\begin{proof}
The second identity in \eqref{eq:defR} can be checked from the differential equation~\eqref{eq:PHnonrenewal}. Let us show the first identity. By Lemma~\ref{lem:brockett}, the stochastic differential equation
\begin{equation*}
dx 
= 
\sum_{\ell, m=1}^p (E^{p+1}_{m\ell} - E^{p+1}_{\ell\ell})x\dN_{S_{\ell m}} 
+ 
\sum_{\ell=1}^{p} (E^{p+1}_{p+1,\ell} - E^{p+1}_{\ell\ell})x\dN_{v_\ell}
\end{equation*}
represents a time-homogeneous Markov process with state space~$\{e_1,
\dotsc, e_{p+1}\} \subset \mathbb{R}^{p+1}$ and the infinitesimal
generator in~\eqref{eq:infgen}. Therefore, identifying $e_{p+1}$, the
absorbing state, with the zero vector, we can see that the equation
\begin{equation}\label{eq:novnoeinfc}
dx 
= 
\sum_{\ell, m=1}^p (E^{p+1}_{m\ell} - E^{p+1}_{\ell\ell})x\dN_{S_{\ell m}} 
+ 
\sum_{\ell=1}^{p} (- E^{p+1}_{\ell\ell})x\dN_{v_\ell}
\end{equation}
represents a time-homogeneous Markov process with state space~$\{e_1,
\dotsc, e_p, 0\} \subset \mathbb{R}^{p+1}$ and the infinitesimal
generator in~\eqref{eq:infgen}. Since in \eqref{eq:novnoeinfc} the
last element of the variable~$x$ plays no role, the equation
\eqref{eq:novnoeinfc} is equivalent to the stochastic differential
equation \eqref{eq:PHnonrenewal} with the state space~$\mathbb{R}^p$.
Thus, \eqref{eq:PHnonrenewal} gives the time-homogeneous Markov process
with the infinitesimal generator \eqref{eq:infgen} and the state
space~$\{e_1, \dotsc, e_p, 0\} \subset \mathbb{R}^{p}$, where
$0 \in \mathbb{R}^p$ is the absorbing state. Since $x(0)$ follows
$e_\phi$, by the definition of phase-type distributions as the exit
time, the random variable $R$ follows $(\phi, S)$.
\end{proof}

Based on the above proved lemma, we can further give a vector
representation of renewal processes whose inter-renewal times have a
phase-type distribution.

\begin{lemma}\label{lem:PHrenewal}
Let $(\phi, S)$ be a phase-type distribution. Let $\epsilon_\phi$ be
the stochastic process that follows $e_\phi$ at every time~$t\geq 0$
independently. Consider the stochastic differential equation
\begin{equation}\label{eq:PHrenewal}
dx 
= 
\sum_{\ell, m=1}^p(E_{m\ell} - E_{\ell\ell}) x\dN_{S_{\ell m}} 
+ 
\sum_{\ell=1}^p (\epsilon_\phi e_\ell^\top - E_{\ell\ell}) x \dN_{v_\ell}
\end{equation}
with the initial condition that $x(0)$ follows $e_\phi$. Define $T_0 = 0$ and
let $0 < T_1 < T_2 < \cdots$ be the times at which $x_\ell = 1$ and the
counter $N_{v_\ell}$ jumps for some $\ell$. Then, the stochastic process $T =
\{T_k\}_{k=0}^\infty$ is the renewal process whose inter-renewal times follow
$(\phi, S)$.
\end{lemma}

\begin{proof}
Let us first show $T_1$ follows $(\phi, S)$. By the definition of $T_1$, on
the interval $[0, T_1)$, the stochastic differential
equation~\eqref{eq:PHrenewal} is equivalent to
\eqref{eq:PHnonrenewal}. In this equivalent equation, the random
variables $T_1$ and $R$ are equal by their definitions and, furthermore,
$R$ follows $(\phi, S)$ by Lemma~\ref{lem:PHdist}. Therefore $T_1$
follows $(\phi, S)$.

At time~$T_1$, the differential equation~\eqref{eq:PHrenewal} reads $dx =
\epsilon_\phi - x$. This means that $x(T_1)$ follows $e_\phi$. Therefore, by
using the memoryless property of Poisson counters, we can show that $T_2-T_1$
follows $(\phi, S)$. In this way, an inductive argument proves that $T$ is a
renewal process with its inter-renewal times following $(\phi, S)$.
\end{proof}

\section{Exponential Transmission Times}\label{sec:inf=exp}

In this section we analyze SIS models under
Assumption~\ref{assm:Tbeta} and give sufficient conditions to achieve the exponential stability of the
infection-free equilibrium. We notice that, under this
assumption, D2 is equivalent to the following condition:
\begin{enumerate}
\item[D2$'$)] Whenever $i$ and $j$ are adjacent, $i$ is susceptible, and
$j$ is infected, the node $i$ becomes infected with an instantaneous
rate of $\beta$.
\end{enumerate}

The next proposition gives a vector representation of the SIS model
under consideration. Throughout the paper, for each $1\leq i\leq
n$ and a real number $\lambda > 0$, we let $N^{(i)}_\lambda$ denote the
Poisson counter with rate $\lambda$. We assume that all the Poisson
counters are independent of each other.

\begin{proposition}\label{prop:PHrecovery}
Let $x^{(i)}$ ($i=1, \dotsc, n$) be the solutions of the stochastic
differential equation:
\begin{equation}\label{eq:SDE:PHrecovery}
\begin{multlined}
dx^{(i)} 
= 
\sum_{\ell, m=1}^p(E_{m\ell} - E_{\ell\ell}) x^{(i)}\dN_{S_{\ell m}}^{(i)} 
\\- 
\sum_{\ell=1}^p E_{\ell\ell} x^{(i)} \dN_{v_\ell}^{(i)}
+ 
\epsilon_\phi(1-\onevtop x^{(i)})\sum_{j=1}^n a_{ij} \onevtop x^{(j)} \dN_\beta^{(j)}
\end{multlined}
\end{equation}
with initial conditions:
\begin{equation}\label{eq:x(0)initcond}
\begin{cases}
\begin{aligned}
&x^{(i)}(0) \text{ follows } e_\phi,\quad i\in \Lambda, 
\\
&x^{(i)}(0) = 0,\quad\text{otherwise.}
\end{aligned}
\end{cases}
\end{equation}
Define the stochastic processes $z_i$ ($i=1, \dotsc, n$) by
\begin{equation}\label{eq:defz}
z_i(t) = 
\begin{cases}
\text{infected},\quad \onevtop x^{(i)}(t) = 1, \\
\text{susceptible},\quad \onevtop x^{(i)}(t) = 0. 
\end{cases}
\end{equation}
Then, the processes $z_i$ are a SIS model satisfying
Assumption~\ref{assm:Tbeta}.
\end{proposition}

\begin{proof}
Let $t_0\geq 0$ and $i$ be arbitrary. First assume that $i$ is
susceptible at time~$t_0$, i.e., $\onevtop x^{(i)}(t_0) = 0$.
Then, from equation~\eqref{eq:SDE:PHrecovery}, we see that
$x^{(i)}(t) = 0$ at least until any of the counters
${N_{\beta}^{(j)}}$ jumps for some $j$ such that $a_{ij} = 1$. Until
that time instant, the differential equation~\eqref{eq:SDE:PHrecovery}
reads
\begin{equation}\label{eq:PHrecovery:hoge}
dx^{(i)} 
= 
\epsilon_\phi \sum_{j=1}^n a_{ij} \onevtop x^{(j)} \dN_\beta^{(j)}. 
\end{equation}
To this inequality we multiply $\onevtop$ from the left and obtain $d(\onevtop
x^{(i)}) = \sum_{j=1}^n a_{ij} \onevtop x^{(j)} \dN_\beta^{(j)}$ because $
\epsilon_\phi$ follows $e_\phi$ and $e_\phi$ is one of the standard unit vectors
with probability one. Therefore, if $a_{ij} = 1$ and $\onevtop x^{(j)} = 1$,
then the quantity $\onevtop x^{(i)}$ becomes one whenever ${N_\beta^{(j)}}$
jumps. In other words, whenever $i$ is susceptible, $j$ is adjacent to $i$, and
$j$ is infected, the node~$i$ becomes infected with the constant rate of
$\beta$. Therefore, D2$'$ is satisfied.

Assume that $i$ becomes infected at time~$t_0$. If $t_0 = 0$,
then $x^{(i)}(t_0)$ follows $e_\phi$ by the initial
condition~\eqref{eq:x(0)initcond}. On the other hand, if $t_0 > 0$,
then the infection must occur by the transmission from a neighboring
node and, by the argument in the last paragraph, such a transmission
occurs when one of its neighbors $j$ is infected and $N_\beta^{(j)}$
jumps. Since multiple Poisson counters cannot jump at the same time
with probability one, the equation \eqref{eq:PHrecovery:hoge} implies
that $x^{(i)}(t_0)$ follows $e_\phi$. Therefore, $x^{(i)}(t_0)$ follows
$e_\phi$ whatever value $t_0$ takes and hence we can, without loss of
generality, assume that $t_0 = 0$. Until $i$ becomes susceptible, that is,
until $x^{(i)}$ becomes $0$, the stochastic differential
equation~\eqref{eq:SDE:PHrecovery} is indeed equivalent to the vector
representation~\eqref{eq:PHnonrenewal} of the distribution~$(\phi,
S)$, which has $0\in \mathbb{R}^p$ as its exit state. This argument
shows that the length of time until $i$ becomes susceptible follows
$(\phi, S)$. Therefore D1 is also satisfied.

Finally, the initial conditions~\eqref{eq:x(0)initcond} ensure that 
D3 is also satisfied. This completes the proof. 
\end{proof}

Using the stochastic differential equations~\eqref{eq:SDE:PHrecovery} we can
derive the following sufficient condition for the exponential stability of the
infection-free equilibrium.

\begin{theorem}\label{thm:PHrecovery}
Under Assumption~\ref{assm:Tbeta}, if the matrix $$\mathcal A_\beta = I \otimes
S^\top + \beta A \otimes (\phi\onevtop)$$ is Hurwitz stable, then the
infection-free equilibrium of the SIS model is exponentially stable.
\end{theorem}

\begin{proof}
Define $\xi^{(i)}(t) = E[x^{(i)}(t)]$. Since $p_i(t) = \onevtop \xi^{(i)}(t)$,
it is sufficient to show that $\xi^{(i)}(t)$ converges to zero exponentially
fast for every $i$. Taking expectations in the differential
equation~\eqref{eq:SDE:PHrecovery} (for details, see,
e.g., \cite{Brockett2009}) yields
\begin{equation*} 
\begin{multlined}[.95\linewidth]
\frac{d\xi^{(i)}}{dt}
= 
\sum_{{\ell, m=1}}^p(E_{m\ell} - E_{\ell\ell}) S_{\ell m} \xi^{(i)}
- 
\sum_{\ell=1}^p E_{\ell\ell} v_\ell\xi^{(i)}
\\
+ 
\phi \onevtop \sum_{j=1}^n a_{ij} \beta \xi^{(j)}
- \phi \onevtop \sum_{j=1}^n a_{ij} \beta E[\onevtop x^{(i)} \onevtop x^{(j)}].  
\end{multlined}
\end{equation*}
Ignoring the last negative term in this equation, we obtain the inequality
${d\xi^{(i)}}/{dt} \leq S^\top \xi^{(i)} + \beta (A_i \otimes (\phi \onevtop))
\xi$, where the $\mathbb{R}^{np}$-valued function $\xi$ is obtained by stacking
$\xi^{(1)}$, $\dotsc$, $\xi^{(n)}$. Therefore, we see that $d\xi/dt \leq
\mathcal A_\beta \xi$. Hence, if $\mathcal A_\beta$ is Hurwitz stable, then the
comparison principle~\cite{Kirkilionis2004} shows that $\xi(t)$ converges to
zero exponentially fast as $t\to \infty$, as desired.
\end{proof}

Using Theorem~\ref{thm:PHrecovery}, we can prove the validity of the condition
\eqref{eq:Cator:PHrecovery} under irreducibility conditions.

\begin{theorem} 
In addition to Assumption~\ref{assm:Tbeta}, assume that the matrices $A$ and $S$
are irreducible and $\phi$ is positive. Then, the condition~\eqref{eq:Cator:PHrecovery} is sufficient for the exponential stability of the
infection-free equilibrium.
\end{theorem}

\begin{proof}
First, we shall see that $\eta(\mathcal A_\beta)$ is strictly increasing
with respect to $\beta$. In fact, for an arbitrary $\epsilon > 0$, we
have $\mathcal A_{\beta+\epsilon} = \mathcal A_{\beta} + \epsilon A\otimes
(\phi \onev^\top)$. In this decomposition, $\mathcal A_{\beta}$ is Metzler
and irreducible by the assumption. Moreover $\epsilon A\otimes (\phi
\onev^\top)$ is nonzero and nonnegative. Thus,
Lemma~\ref{lem:eta<...} shows $\eta(\mathcal A_\beta) <
\eta(\mathcal A_{\beta+\epsilon})$.

Therefore, to prove the given claim, it is sufficient to show that $\eta(\mathcal A_{\beta_0}) = 0$ where $\beta_0 = 1/(E[R]\eta(A))$. This is equivalent to the existence of a positive null vector for $\mathcal A_{\beta_0}$ by Lemma~\ref{lem:metzler} because $\mathcal A$ is Metzler and irreducible. In the rest of the proof, we shall show that $z = u \otimes {(-(S^{-1})^\top \phi)}$ is such a null vector, where $u$ is the eigenvector of $A$ corresponding to the eigenvalue $\eta(A)$. Let us first show that $z$ is positive. The vector $u$ can be taken to be positive by the Perron-Frobenius theory because $A$ is irreducible (see \cite{Vandergraft1968}). Moreover, since $S$ is Metzler and Hurwitz stable, Lemma~\ref{lem:Metzler:Hurwitz} shows that $S^{-1}$ does not have a positive entry. Also $S^{-1}$ clearly does not have a zero row. Therefore, since $\phi$ is positive, the vector~$-(S^{-1})^\top \phi$ is also positive. Hence $z$ is indeed positive. Now, let us compute the product $\mathcal A_{\beta_0} z$:
\begin{equation}\label{eq:computecalAbeta0}
\begin{aligned}
\mathcal A_{\beta_0} z
&=
-u \otimes \phi - (\beta_0 \eta(A) u) \otimes  (\phi \onevtop (S^{-1})^\top\phi)
\\
&=
-(u \otimes \phi) (
1 + \beta_0 \eta(A)\onevtop (S^{-1})^\top\phi). 
\end{aligned}
\end{equation}
Since the mean of the distribution $(\phi, S)$ equals $-\phi^\top
S^{-1}\onev$ (\cite{Asmussen1996}), we have $\onevtop
(S^{-1})^\top\phi =  \phi^\top S^{-1} \onev = -E[R]$. Therefore, by
the definition of $\beta_0$, the equation \eqref{eq:computecalAbeta0}
shows $\mathcal A_{\beta_0} z = 0$, as desired. This completes the
proof of the theorem.
\end{proof}

\begin{remark}
An advantage of Theorem~\ref{thm:PHrecovery} over the condition
\eqref{eq:Cator:PHrecovery} is that the theorem explicitly gives an
upper bound on the speed of convergence to the infection-free state as
the maximum real part of the eigenvalues of~$\mathcal A_\beta$. This
would enable us to, for example, design the optimal strategies for
distributing preventive resources over networks under constraints on
the speed and the total amount of the resources available, as in
\cite{Preciado2014}.
\end{remark}

\section{Exponential Recovery Times}\label{sec:cure=exp}

As the dual of the previous section, in this section we analyze SIS
models under Assumption~\ref{assm:Rdelta}. Under this assumption, D1
is equivalent to the following condition:
\begin{enumerate}
\item[D1$'$)] Whenever a node becomes infected, it will recover with the
instantaneous rate of $\delta$.
\end{enumerate}

The next proposition gives a vector representation of the
corresponding SIS model as in Proposition~\ref{prop:PHrecovery}.

\begin{proposition}\label{prop:PHinfection}
Let $x^{(i)}$ ($i=1, \dotsc, n$) be the solutions of the stochastic
differential equation
\begin{equation}\label{eq:SDE:PHinfection}
\begin{multlined}
dx^{(i)}  
=  -x^{(i)}  \dN_\delta^{(i)} 
+\,\sum_{\mathclap{\ell, m=1}}^p\,
(E_{m\ell} - E_{\ell\ell})x^{(i)}  \dN_{S_{\ell m}}^{(i)} +
\\
\sum_{\mathclap{\ell=1}}^p (\epsilon_\phi e_\ell^\top  - E_{\ell\ell}) x^{(i)}  \dN^{(i)}_{v_\ell} +
  \epsilon_\phi (1-\onevtop x^{(i)} )\sum_{\mathclap{j=1}}^n a_{ij} 
\sum_{\mathclap{\ell=1}}^p x_\ell^{(j)} \dN_{v_\ell}^{(j)}
\end{multlined}
\end{equation}
with the initial conditions~\eqref{eq:x(0)initcond}. Then the stochastic processes $z_1$, $\dotsc$, $z_n$ defined by \eqref{eq:defz} are a SIS model satisfying Assumption~\ref{assm:Rdelta}.
\end{proposition}

\begin{proof}
Let $t_0\geq 0$ and $i$ be arbitrary. First assume that $i$ is susceptible at time~$t_0$. After time~$t_0$, and while $i$ is susceptible, the differential equation~\eqref{eq:SDE:PHinfection} reads $dx^{(i)}   = \epsilon_\phi \sum_{j=1}^n a_{ij} \sum_{\ell=1}^p x_\ell^{(j)} \dN^{(j)}_{v_\ell}$. Multiplying $\onevtop$ from the left yields $d(\onevtop x^{(i)} )  = \sum_{j=1}^n a_{ij} \sum_{\ell=1}^p x_\ell^{(j)} \dN^{(j)}_{v_\ell}$. Therefore, if $a_{ij} = 1$ and $x_\ell^{(j)} = 1$, then, whenever the counter $N_{v_\ell}^{(j)}$ jumps, the quantity $\onevtop x^{(i)} $ becomes one, that is, $i$ becomes infected.

We then consider the case that $i$ becomes infected at time~$t_0$. As in the proof of Theorem~\ref{thm:PHrecovery}, without loss of generality we can assume $t_0 = 0$. After time~$0$, and while $i$ is infected, the differential equation~\eqref{eq:SDE:PHinfection} reads
\begin{equation*} 
\begin{multlined}[.9\linewidth]
dx^{(i)}  
= 
-x^{(i)}  \dN_\delta^{(i)}
+ 
\,\sum_{{\ell, m=1}}^p \,
(E_{m\ell} - E_{\ell\ell})x^{(i)}  \dN^{(i)}_{S_{\ell m}}
+ 
\\
\sum_{\ell=1}^p (\epsilon_\phi e_\ell^\top - E_{\ell\ell}) x^{(i)}  \dN^{(i)}_{v_\ell}.
\end{multlined}
\end{equation*}
By the first term of this equation, we see that $x^{(i)}$ becomes zero when and
only when the counter~$N_\delta^{(i)}$ jumps. This implies that $i$ recovers
with a rate of~$\delta$ and hence shows D1$'$ to be true. On the other hand,
until $N_\delta^{(i)}$ jumps, the variable~$x^{(i)}$ follows the same
differential equation as~\eqref{eq:PHrenewal}. Since $x^{(i)} (0)$ follows
$e_\phi$, by Lemma~\ref{lem:PHrenewal}, the times at which $x^{(i)} _\ell = 1$
and $N_{v_\ell}^{(i)}$ jumps for some $\ell$ form the renewal process with its
inter-renewal times following $(\phi, S)$. This observation and the argument in
the first paragraph of this proof prove that the stochastic processes $z_i$
satisfy D2. Also D3 holds true by the initial
conditions~\eqref{eq:x(0)initcond}.
\end{proof}

From Proposition~\ref{prop:PHinfection} we obtain the following
criterion for the exponential stability:

\begin{theorem}\label{thm:PHinfection}
Under Assumption~\ref{assm:Rdelta}, if
\begin{equation}\label{eq:delta>...}
\delta > \eta\bigl(I\otimes S^\top + (A+I)\otimes (\phi v^\top) \bigr), 
\end{equation}
then the infection-free equilibrium of the SIS model is exponentially stable. 
\end{theorem}

\begin{proof}
As in the proof of Theorem~\ref{thm:PHinfection}, it is sufficient to show that
$\xi^{(i)}(t)$ converges to zero exponentially fast as $t\to\infty$ for every
$i$. Taking the expectation in \eqref{eq:SDE:PHinfection}, we can see that
\begin{equation*}
\begin{aligned}
\!\!\frac{d\xi^{(i)}\!\!}{dt}
&= 
-\delta \xi^{(i)}\! + 
\sum_{\mathclap{\ell, m=1}}^p(E_{m\ell} - E_{\ell\ell}) S_{\ell m} \xi^{(i)}
+
\sum_{\ell=1}^p (\phi e_\ell^\top - E_{\ell\ell}) \xi^{(i)} v_\ell
\\
&\hspace{.5cm}+\phi \sum_{j=1}^n a_{ij}\sum_{\ell=1}^p e_\ell^\top \xi^{(j)} v_\ell - \phi \onevtop \sum_{j=1}^n \sum_{\ell=1}^p a_{ij} E[x^{(i)} x^{(j)}_\ell] v_\ell
\\
&\leq
(-\delta I + S^\top + \phi v^\top)\xi^{(i)} + (A_i\otimes (\phi v^\top)) \xi.
\end{aligned}
\end{equation*}
Therefore ${d\xi}/{dt} \leq ( I\otimes S^\top + (A+I)\otimes (\phi v^\top) -
\delta I ) \xi$. Hence, for $p_i(t)$ to converge to zero exponentially fast, it
is sufficient that the matrix~$I\otimes S^\top + (A+I)\otimes (\phi v^\top) -
\delta I$ is Hurwitz stable by the same argument as in the proof of
Theorem~\ref{thm:PHrecovery}. This proves the sufficiency of the
condition~\eqref{eq:delta>...}.
\end{proof}

Using Theorem~\ref{thm:PHinfection}, we can then validate the effectiveness of
the condition \eqref{eq:Cator:PHinfection} under irreducibility conditions.

\begin{theorem}
Under Assumption~\ref{assm:Rdelta}, if $A$ is irreducible and $v$ is positive,
then \eqref{eq:Cator:PHinfection} implies the exponential stability of the
infection-free equilibrium.
\end{theorem}

\begin{proof}
Assume \eqref{eq:Cator:PHinfection}. Then we have $f(\delta) < 1/(1+\eta(A))$.
Since the probability density function of $(\phi, S)$ has the form
$\phi^\top \exp(St)v$ ($t\geq 0$), we obtain
\begin{equation}\label{eq:usedsome}
\phi^\top (\delta I - S)^{-1} v < 1/(1+\eta(A)). 
\end{equation}
By Theorem~\ref{thm:PHinfection}, it is sufficient to show that the Metzler
matrix  $\mathcal B = I\otimes S^\top + (A+I)\otimes (\phi v^\top) - \delta I$
is Hurwitz stable. Define $z = u \otimes ((\delta I - S)^{-1}v)$, where $u > 0$
is the Perron-Frobenius eigenvector of $A$. Since $S-\delta I$ is Hurwitz
stable, the inverse $(\delta I -S)^{-1} = -(S - \delta I)^{-1}$ is nonnegative
and does not have a zero row by Lemma~\ref{lem:Metzler:Hurwitz}. Therefore the
product $(\delta I - S)^{-1}v$ is positive because $v$ is positive. Hence we see
that $z$ is positive. Now, using \eqref{eq:usedsome}, we can actually show that
$\mathcal B z < 0$. 
Therefore, by Lemma~\ref{lem:Metzler:Hurwitz}, $\mathcal B$ is Hurwitz stable.
\end{proof}

\section{Numerical Simulations}\label{sec:num}

In this section, we illustrate Theorem~\ref{thm:PHinfection} through the
comparison with the condition obtained in \cite{Cator2013a}. Let $R$ follow a
Weibull distributions with probability density function
$({\alpha}/{b})({t}/{b})^{\alpha-1}\exp(-(t/b)^\alpha)$, $t\geq 0$, where
$\alpha$ and $b$ are positive parameters. In order to normalize the mean of the
distributions to be one, we fix $b = \Gamma(1+\alpha^{-1})$ where
$\Gamma(\cdot)$ denotes the Gamma function. In \cite{Cator2013a}, based on the
condition \ref{eq:Cator:PHinfection}, it is concluded that the infimum of the
recovery rate $\delta$ such that the SIS model under Assumption
\ref{assm:Rdelta} has the infection-free steady state equals $\delta_0 =
\Gamma(1+\alpha^{-1}) \Gamma(\alpha+1)^{1/\alpha} \eta(A)^{1/\alpha}$.

\newcommand{\mynum}{.88}

We compare the quantity $\delta_0$ with the infimum recovery rate $\delta_1 =
\eta(I\otimes S^\top + (A+I)\otimes (\phi v^\top))$, which we can obtain from
Theorem~\ref{thm:PHinfection}. Let $\mathcal G$ be a realization of the
Erd\H{o}s-R\'enyi graph with $500$ nodes. We vary the parameter $\alpha$ of the
Weibull distribution as $\alpha = 0.5$, $1$, $\dotsc$, $4.5$, and $5$. The
Weibull distributions are fitted with phase-type distributions using the
expectation\nobreakdash-maximization algorithm proposed in~\cite{Asmussen1996}
(and available
at~\mbox{\small\url{http://home.math.au.dk/asmus/pspapers.html}}). Some of the
fitting results are shown in Fig.~\ref{fig:weibullFitting}. In
Fig.~\ref{fig:weibullInfection}, we compare the two recovery rates $\delta_0$
and $\delta_1$.
\begin{figure}[tb]
\vspace{.1cm}
\centering \includegraphics[width=\mynum\linewidth]{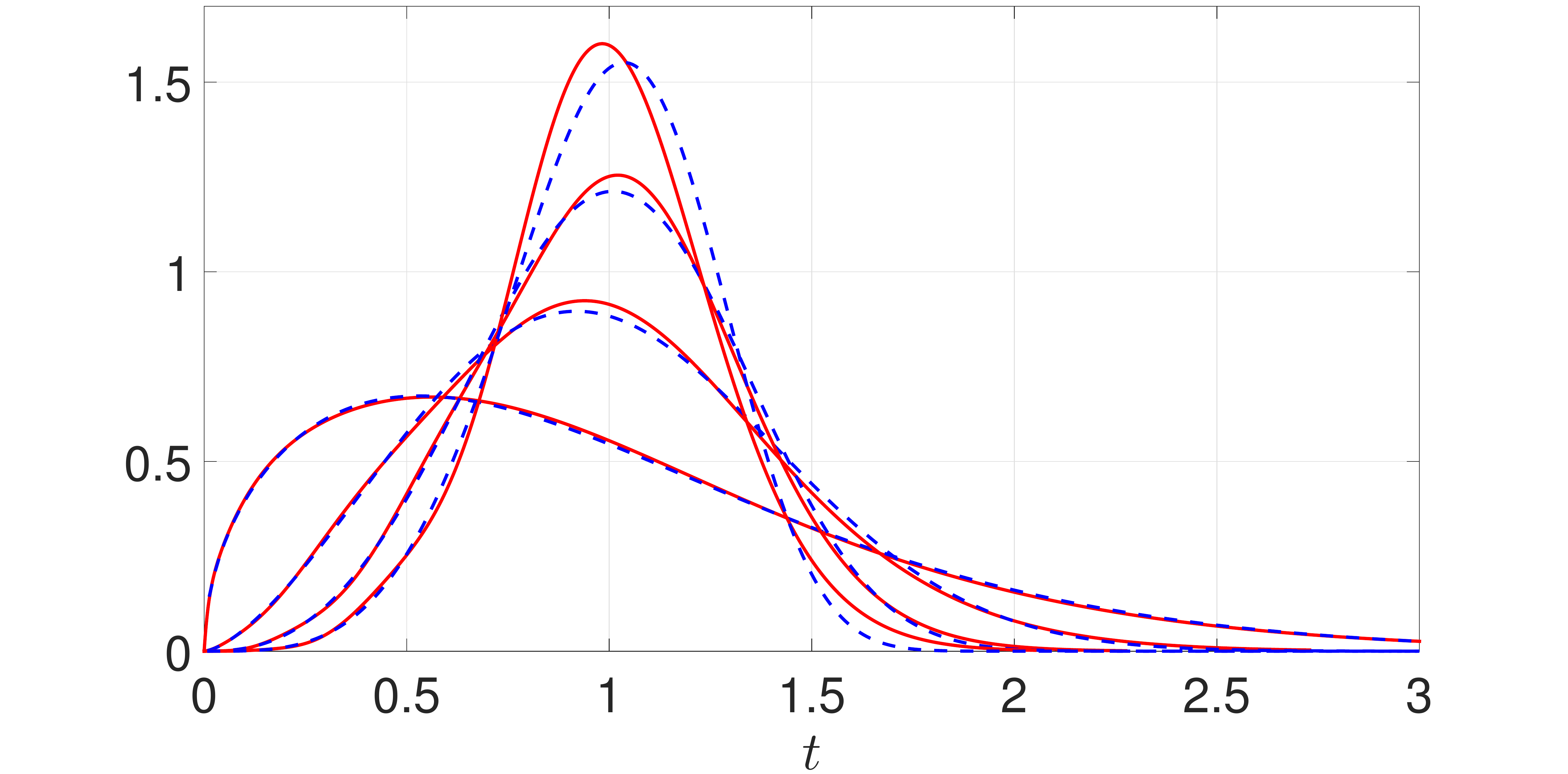}
\caption{Weibull distributions fitted by phase-type distributions for
$\alpha = 1.5$, $2.5$, $3.5$, and $4.5$. Solid:  phase-type distributions. Dashed:  Weibull distributions.} \label{fig:weibullFitting}
\vspace{.25cm}
\includegraphics[width=\mynum\linewidth]{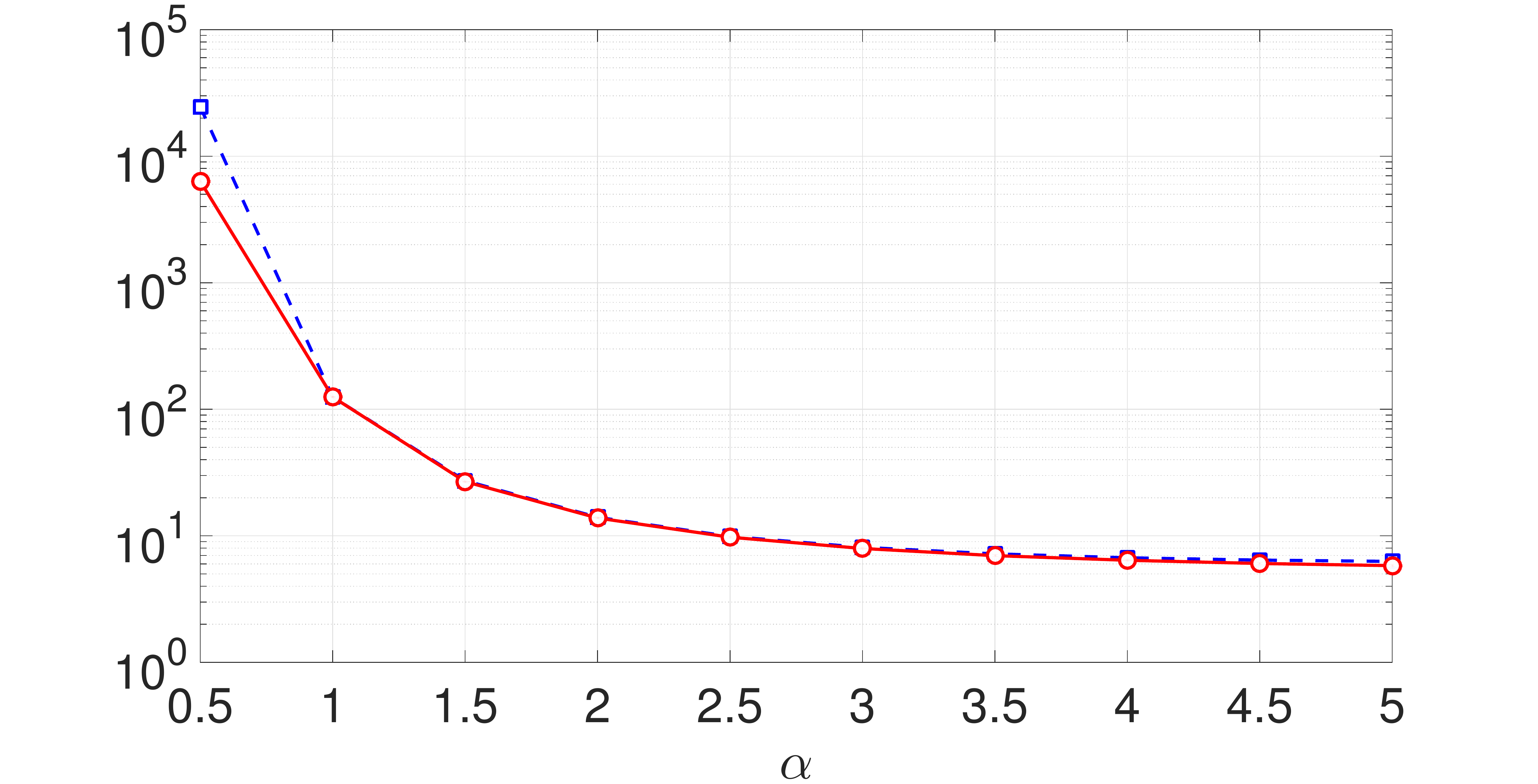}
\caption{Minimum recovery rates. Circle: $\delta_1$. Square: $\delta_0$.}
\label{fig:weibullInfection}
\end{figure}
The graph shows the coincidence of the two rates except at $\alpha =
0.5$.

\section{Conclusion}

In this paper we have analyzed SIS models of spreading over networks with phase-type transmission and recovery times. We have derived sufficient conditions to tame the spread in terms of the eigenvalues of matrices that depend on both the graph structure and the parameters of the phase-type distribution. Our results mathematically justify the conditions found in \cite{Cator2013a} without using asymptotic arguments. The generality of the approach herein introduced is supported by the fact that the set of phase-type distributions is dense in the set of all the positive random variables. As a future work, we will develop control strategies to contain epidemic spreading with phase-type rates, as well as applications in the context of distribution of online content in social networks.


\end{document}